\documentclass[runningheads,a4paper]{cmmr2023}
\usepackage{times}

\usepackage{amssymb}
\usepackage{amsmath}
\usepackage{graphicx}
\usepackage[hyphens]{url}
\usepackage{xspace}
\usepackage{musicography}
\usepackage{amssymb}
\usepackage{bbold}
\def\ie{i.e.,\xspace}
\def\eg{e.g.,\xspace}
\def\wrt{w.r.t.\xspace}
\newcommand{\keywords}[1]{\par\addvspace\baselineskip
\noindent\keywordname\enspace\ignorespaces#1}

\DeclareUnicodeCharacter{0301}{*************************************}

\usepackage{xcolor}

\makeatletter
\newcommand{\oset}[3][0ex]{%
  \mathrel{\mathop{#3}\limits^{
    \vbox to#1{\kern-2\ex@
    \hbox{$\scriptstyle#2$}\vss}}}}
\makeatother

\def\ie{i.e.,\xspace}
\def\eg{e.g.,\xspace}
\def\wrt{w.r.t.\xspace}

\newcommand{\finexsymbol}{\ensuremath\Diamond}
\newcommand{\finex}{\hfill$\finexsymbol$}

\newcommand{\ASD}{\mathcal{A}}
\newcommand{\RSD}{\mathcal{R}}

\def\<#1>{\langle #1 \rangle}

\def\Music21{\textsf{Music21}\xspace}

\def\ps13{\textsf{\small ps13}\xspace}

\newcommand{\dom}{\ensuremath{\mathit{dom}}}

\newcommand{\Semiring}{\mathcal{S}}

\newcommand{\zero}{\mathbb{0}}
\newcommand{\one}{\mathbb{1}}

\def\Reduction#1#2#3#4{%
\mathrel{\raise1.0ex\hbox{%
\vtop{\ialign{##\crcr%
\raise0.0ex\hbox{$\hfil\scriptstyle{\ #1\ }\hfil$}\crcr%
\noalign{\nointerlineskip}%
\rightarrowfill\crcr%
\noalign{\nointerlineskip}%
\raise0.0ex\hbox{$\hfil\scriptstyle{\ #2\ }\hfil$}\crcr}}}{}^{#3}_{#4}}}
\def\Leduction#1#2#3#4{%
\mathrel{\raise1.0ex\hbox{%
\vtop{\ialign{##\crcr%
\raise0.0ex\hbox{$\hfil\scriptstyle{\ #1\ }\hfil$}\crcr%
\noalign{\nointerlineskip}%
\leftarrowfill\crcr%
\noalign{\nointerlineskip}%
$\hfil\scriptstyle{\ #2\ }\hfil$\crcr}}}{}^{#3}_{#4}}}
\def\hookReduction#1#2#3#4{%
\mathrel{\raise1.2ex\hbox{%
\vtop{\ialign{##\crcr%
\raise0.0ex\hbox{$\hfil\scriptstyle{\ #1\ }\hfil$}\crcr%
\noalign{\nointerlineskip}%
$\lhook\joinrel$\hspace{-0.35em}
\rightarrowfill\crcr%
\noalign{\nointerlineskip}%
$\hfil\scriptstyle{\ #2\ }\hfil$\crcr}}}{}^{#3}_{#4}}}
\def\hoookReduction#1#2#3#4{%
\lhook\joinrel\hspace{-0.50em}
\raise0.85ex\hbox{%
\vtop{\ialign{##\crcr%
\raise0.4ex\hbox{$\hfil\scriptstyle{\ #1\ }\hfil$}\crcr%
\noalign{\nointerlineskip}%
\rightarrowfill\crcr%
\noalign{\nointerlineskip}%
$\hfil\scriptstyle{\ #2\ }\hfil$\crcr}}}{}^{#3}_{#4}}

\def\frew#1#2#3#4#5#6#7#8{
\setbox0=\hbox{$#6 #7 #1 #8$}%
\setbox1=\hbox{$#6 #7 #2 #8$}%
\ifdim \wd0>\wd1 \rlap{\rlap{\hbox to \wd0{#5}}%
                            {\hbox to\wd0{\hfil\lower #3\box1\relax\hfil}}}{\raise #4\box0}%
\else \rlap{\rlap{\hbox to \wd1{#5}}{\hbox to\wd1{\hfil\raise #4\box0\relax\hfil}}}{\lower #3\box1}%
\fi
}

\usepackage[
backend=biber,
style=alphabetic,
sorting=ynt
]{biblatex}

\begin{document}

\mainmatter  

\title{8+8=4: Formalizing Time Units \\ to Handle Symbolic Music Durations}

\titlerunning{8+8=4: Formalizing Music Durations}

%
%
\author{Emmanouil Karystinaios\inst{1}\and Francesco Foscarin\inst{1} \and Florent Jacquemard\inst{2} \and Masahiko Sakai\inst{3}  \and Satoshi Tojo\inst{4} \and  Gerhard Widmer\inst{1,5}}
%
\authorrunning{Karystinaios et al.}

\institute{Institute of Computational Perception, Johannes Kepler University Linz, Austria \and CNAM Paris, France \and Nagoya University, Japan \and Asia University, Japan \and LIT AI Lab, Linz Institute of Technology, Austria\\ \email{emmanouil.karystinaios@jku.at} \\ \email{francesco.foscarin@jku.at}}

%

\maketitle

\begin{abstract}
This paper focuses on the nominal durations
of musical events (notes and rests) in a symbolic musical score, and on how to conveniently handle these in computer applications. We propose the usage of a temporal unit that is directly related to the graphical symbols in musical scores, and pair this with a set of operations that cover typical computations in music applications. We formalise this time unit and the more commonly used approach in a single mathematical framework, as semirings, algebraic structures that enable an abstract description of algorithms / processing pipelines.
We then discuss some practical use cases and  highlight when our system can improve such pipelines by making them more efficient in terms of data type used and the number of computations.

\end{abstract}

\section{Introduction}
In a musical score, the duration of musical events (\ie notes and rests) is defined by a finite set of symbols, and their temporal position by summing the duration of the previous musical events.
Computer applications that deal with musical scores typically work with \textit{Relative Symbolic Duration} (RSD) units, i.e., they choose a reference note duration and model all temporal information as ratios of that reference. For example, for the first four notes of the upper voice in Figure~\ref{fig:score_example}, one can choose a quarter note \musQuarter{} as a reference and represent the durations in the first two beats as the sequence $[0.5,0.25,0.25, 1]$.
This kind of encoding shows its limits for certain durations, typically those produced by irregular groupings (also called tuplets). The 5th note in the top voice in the figure would have a duration of $2/3$, which is a periodic number not representable as a floating point value in computer applications, thus requiring a truncation. This introduces an error that propagates to all subsequent musical events and creates a number of problems for applications that require exact matching of temporal positions.


%

Two main approaches have been proposed to solve this problem. The first is the \textit{fraction approach}, implemented, for example, by the Python library Music21~\cite{music21}. It involves the representation of durations with specific Python objects made to mimic a fraction. This eliminates the rounding problems, but the fraction object is inefficient to handle with respect to native Python types and is not supported by libraries for heavy computations such as Numba, Pytorch, or TensorFlow.
The second method, the \textit{common divisor} approach, consists of setting the aforementioned reference duration to a value that is a common divisor of all durations appearing in a given piece or set of pieces. All temporal information can then be expressed with natural numbers, enabling very efficient computations. This solution is adopted by the Python library Partitura~\cite{partitura_mec}, in some musical score storage formats such as MIDI, MEI, MusicXML, and in other computer music frameworks (e.g.,~\cite{foscarin2019parse}). However, this solution is still problematic for real-time scenarios when we do not know all duration in advance or when the piece can be modified. When a new duration is added that is not a multiple of the reference, the reference must be recomputed and all values updated.

\begin{figure}
    \centering
    \includegraphics[width=0.65\columnwidth]{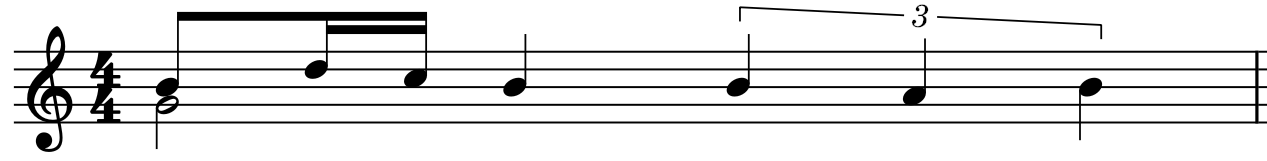}
    \caption{A musical score example with two problematic configurations: a tuplet and an incomplete bottom voice.
    }
    \label{fig:score_example}
\end{figure}

\begin{example}\label{ex:running}
Let us consider a toy application on the score of Figure~\ref{fig:score_example}. We are interested in importing it from an MEI file, splitting the third note (the C) in the top voice into two, and producing a pianoroll representation.
The notes in the top voice have durations, in RSD units (with a quarter note as reference),  of  $[\frac{1}{2},\frac{1}{4},\frac{1}{4},1,\frac{2}{3},\frac{2}{3}, \frac{2}{3}]$. In the common division approach, we first need to 
compute as reference value $\delta$ a common divisor of all absolute note durations, the largest (i.e., the greatest common divisor GCD)
if we want to optimize memory usage. In this case, this is $\frac{1}{12}$ of a quarter note. We then express each duration as a multiple of the reference, i.e., $[6,3,3,12,8,8,8]$. If we want to split the third note, we need to recalculate the value of $\delta$ as $1/24$, update the durations to $[12,6,6,24,16,16,16]$, and finally split the third note in two notes with duration $3$. 
We can produce the pianorolls of the two voices independently and then perform an element-wise sum to obtain the score pianoroll.
However, the second voice is logically incomplete in the score, missing an explicit half-note rest.\footnote{Ideally, both voices will have the same duration, but in real scores, this is often not the case; see \cite{asap-dataset,foscarin2021data} for a discussion about score quality.} Thus, we first need to compute the maximum between the total duration of the two voices and insert the missing rests in the second voice. We compute the onset of each note by summing the durations of all previous notes in the same voice. 
\finex
\end{example}

This paper discusses an alternative approach to handling durations: the use of \textit{Absolute Symbolic Duration} (ASD) units. 
The core idea is to consider the integers implied by the names of the graphical symbols. For example, a  quarter note \musQuarter{} as 4, an eight note \musEighth{} as 8, a 16th note \musSixteenth{} as 16, and so on. Durations produced from irregular groupings are also expressed as integers (see Table~\ref{tab:duration_table}). 
ASD units are already used by the Humdrum **kern file format, and (in a mixed representation with the divs approach) by MEI and MusicXML. However, they are only used to encode single note/rest durations. The typical pipeline procedure is to translate this duration format into relative symbolic durations, as a preprocessing step before any other operation.
\begin{table}[t]
\begin{center}
\begin{tabular}{lccccccc}
\hline
Music Symbol      & \musWhole{}  & \musHalf{}  & \musQuarter{} & $\; \oset[1.8ex]{3}{\text{\musQuarter{}}}\; $  & \musEighth{}  & $\overset{3}{\text{\musEighth{}}}$ & \musSixteenth{}  \\
Relative Symbolic Duration (1 = \musQuarter{} ) & 4  & 2 &  1 & $0.\bar{6}$ & 0.5 & $0.\bar{3}$ & 0.25 \\
Absolute Symbolic Duration     & 1  & 2 & 4 & 6 & 8 & 12 & 16 \\
\hline
\end{tabular}
\end{center}
\label{tab:duration_table}
\caption{Examples of music symbols and corresponding durations in RSD and ASD units. Notes with ``3'' on top are notes that are part of a triplet.}
\end{table}

On the contrary, we explore the usage of ASD, as ``standalone'' units to manipulate musical score durations. To make this practicable, we define two operations that cover typical use cases and we prove that, like RSD, ASD units form a \textit{semiring}, an algebraic structure that enables a more abstract general description of processing pipelines.
The actual computations can later be performed in ASD or RSD (or a mixture of the two), depending on the situation. This is enabled by an isomorphism that we provide to translate between the two units.
Finally, we discuss some practical cases where one unit is to be preferred over the other to make the pipeline more efficient in terms of the number of operations and data types that are considered. We implement some algorithms that use ASD units in the Python library Partitura~\cite{partitura_mec}.


\section{Definitions}
In this section, we first introduce the semiring; then we formally define \textit{ASD} and \textit{RSD} units and a morphism between them. Our goal with the introduction of this formalism is to give a general, abstract way of describing algorithms on music durations which is valid for both ASD and RSD units. Such algorithms can practically be performed in one unit or the other (or a mix of the two) depending on the specific application (see Section~\ref{sec:practical_usage}). 

\subsection{Semiring}
Formally, a \emph{semiring} 
$\Semiring = \< \mathbb{S}, \oplus, \otimes>$ 
is an algebraic structure that consists of a domain $\mathbb{S} = \dom(\Semiring)$, and two  associative binary operators~$\oplus$ and $\otimes$. Some properties must be verified: $\oplus$ is commutative, and $\otimes$ distributes over~$\oplus$, i.e.,  $\forall x, y, z \in \mathbb{S}$, $x \otimes (y \oplus z) = (x \otimes y) \oplus (x \otimes z)$.

Note that there is no complete agreement in the literature over the exact definition of a semiring. Other research (e.g., \cite{pin1998tropical}) defines the two operations of a semiring with a neutral element ($\zero$ and $\one$ respectively), such that $\zero$ is absorbing for~$\otimes$: $\forall x\in \mathbb{S}$, $\zero \otimes x = x \otimes \zero = \zero$. Then the semiring without neutral elements is called \textit{hemiring}.
However, similarly to~\cite{dudek2010characterizations,sen2021some}, we just use the term semiring, without including neutral elements (and in particular the absorbing propriety of $\zero$). The motivation is that verifying the absorbing propriety requires changes that would take our framework further away from its use for practically useful operations on music duration (more on this in Section~\ref{sec:rsd}). 
Components of any semiring~$\Semiring$ may be superscripted by~$\Semiring$ when needed. 
By abuse of notation, we write $x \in \Semiring$ to denote $x \in \mathbb{S}$.


A semiring $\Semiring$ is \emph{commutative} if $\otimes$ is commutative.
It  is \emph{idempotent} if
for all $x \in \Semiring$, $x \oplus x = x$.
It is \emph{monotonic} \wrt a partial ordering~$\leq$ 
iff for all $x, y, z$,  $x \leq y$ implies
$x \oplus z \leq y \oplus z$,
$x \otimes z \leq y \otimes z$
and $z \otimes x \leq z \otimes y$.
Every idempotent semiring $\Semiring$ induces 
a partial ordering~$\leq_\Semiring$ 
called the \emph{natural ordering} of~$\Semiring$
and defined by: 
for all $x$ and $y$,
$x \leq_\Semiring y \;\mbox{iff}\; x \oplus y = x$.
It holds then that $\Semiring$ is {monotonic} \wrt~$\leq_\Semiring$.
$\Semiring$ is called \emph{total} if
it is idempotent and~$\leq_\Semiring$ is total, 
\ie when for all $x$ and $y$, either $x \oplus y = x$ or $x \oplus y = y$.

Given the particular algebraic properties above, 
semirings can be used as a weight domain  
for optimization problems such as the search for shortest paths in 
weighted graphs or hypergraphs~\cite{Mohri02semiring,Huang08advanceddynamic}.
Indeed, the theory of semirings and in particular the min-plus and max-plus 
Tropical Algebras~\cite{gaubert1997methods} 
is commonly applied in decision theory and operational research, 
performance evaluation and control of dynamic systems,
and also formal language theory,
for quantitative extensions 
of formal computation models~\cite{DrosteKuich09semirings}
(weighted automata and grammars).
They have also been recently used 
for the formalization of musical elements, 
\eg 
harmonic/melodic intervals by Albini~\cite{10.1007/978-3-030-21392-3_6}, and to describe algorithms for musical tasks, \eg 
music transcription by $n$-best parsing~\cite{foscarin2019parse}, 
and melodic distance computation~\cite{Giraud22IC}.

This work focuses on formalisations of musical duration that form idempotent 
and commutative semirings.
Intuitively, in the applications presented in this paper, $\oplus$ selects the longest duration and $\otimes$ aggregates two durations in a single one.

\subsection{Absolute Symbolic Durations}
Let us define the semiring of Absolute Music Duration units $\ASD = \< \mathbb{Q}^+\cup\{ \infty \}, \oplus, \otimes>$ by detailing its domain and the two operations.

\subsubsection{The domain}\hfill\break
The domain $dom(\ASD) = \mathbb{Q}^+\cup\{ \infty \}$ of $\ASD$ contains (but is not limited to) non-null integers implied by the graphical symbol of notes and rests, e.g., quarter notes, eight notes, 16th notes, 32th notes, etc. Intuitively, larger values correspond to shorter notes. The limiting case is the null musical duration (used, for example, for grace notes), which is denoted by $\infty$. 
$dom(\ASD)$ also includes other values that can result from the use of \textit{duration modifiers} in the musical score, such as dots and tuplets, and will be described later in this section. 
We define $\prec^{\ASD}$ to be the strict order of absolute musical durations on the domain of $\ASD$. Elements of $\ASD$ are defined such that $\forall a,b \in dom(\ASD), a \prec^{\ASD} b \iff a > b$. 

\subsubsection{Operations}\hfill\break
We are interested in two operations: a \textit{selection} operation to find the longest duration, and a \textit{concatenation} to combine two or more musical durations. We define $\oplus^{\ASD}$ such that $a \oplus^{\ASD} b = \textrm{min}(a, b)$, as the selection operation. Practically, this operation can be used to select the longest voice within a measure, when their durations do not correspond, like in Example~\ref{fig:score_example}.

The concatenation operation $\otimes^{\ASD}$ is defined as $a \otimes^{\ASD} b = \frac{ab}{a+b}$. This operation expresses mathematically the well-known musical rules about aggregating durations. For example, the concatenation of two eighth notes yields a quarter note, which in our framework can be written as $8 \otimes^{\ASD} 8 = 4$. A more advanced usage for ties and dots is also exemplified in Section~\ref{sec:general}. Readers who are not familiar with the semiring formalisms may find confusing that this concatenation operation, which looks very much like a sum, is denoted with the symbol $\otimes$, but this is what is commonly used and we keep it for consistency.

To prove that $\ASD$ is a semiring we need to prove that we have closure for both operations and that the multiplication distributes over additions. We go slightly further than proving closure and prove that both operations are commutative monoids (i.e. that they are commutative, associative, and there is an identity element) since this could be useful for further extension of our framework. Remember that for simplicity we write $x \in \ASD$ to denote $x \in \dom(\ASD)$.


\begin{lemma}\label{com_plus}
    $\<dom(\ASD), \oplus^{\ASD}>$ is a commutative monoid.
\end{lemma}

\begin{proof}
    Let $a, b, c \in \ASD$. By definition $a \oplus^{\ASD} b = \textrm{min}(a, b)$.
    
    Then from the commutativity and associativity properties of the $\textrm{min}$ operation, $\oplus$ is also commutative and associative. The closure is trivial for $\textrm{min}$. The identity element is $\infty$, i.e., $\forall a \in \ASD, a \oplus \infty = a$.
\end{proof}

\begin{lemma}\label{com_times}
$\<dom(\ASD), \otimes^{\ASD}>$ is a commutative monoid.
\end{lemma}

\begin{proof}
Let $a, b, c \in \ASD$. By definition $a \otimes^{\ASD} b = \frac{ab}{a+b}$. 
By the commutativity of addition and multiplication, it follows that $\otimes^{\ASD}$ is also commutative and associative. Closure is also verified for the same reason.  
Let us investigate if the relationship also holds for the case of the null durations, i.e. $\infty$. We define $a \otimes^{\ASD} \infty$ as the limit $\lim_{b \to \infty} \left(a \otimes^{\ASD} b  \right)$.
\begin{multline*}
    a \otimes^{\ASD} \infty = \lim_{b \to \infty} \left(a \otimes^{\ASD} b  \right) = \lim_{b \to \infty} \left( \frac{ab}{a+b}\right) = a \lim_{b \to \infty}\left(\frac{b}{a+b}\right)= a \lim_{b \to \infty}\left(\frac{1}{\frac{a}{b}+1}\right) = \\ 
    = a \left( \frac{\lim_{b \to \infty}1}{\lim_{b \to \infty}\frac{a}{b}+1} \right) = a \left(\frac{1}{\lim_{b \to \infty} \frac{a}{b} + 1} \right) = a \left(\frac{1}{0+1}\right) = a     
\end{multline*}
\noindent
Since $\otimes^{\ASD}$ is commutative, this also holds for the case $\infty \otimes^{\ASD} a$. We also proved that $\infty$ is the neutral element of $\otimes^{\ASD}$.

\finex
\end{proof}

We will now prove some Lemmas that will be useful for the proof of Theorem~\ref{th:semiring}.
\begin{lemma}\label{inequality}
    Let $a,b,c \in \ASD$, then $b<c \iff a\otimes^\ASD b < a\otimes^\ASD c$.
\end{lemma}
\begin{proof}
    \begin{multline*}
        b<c \overset{a>0}{\iff} ab < ac \overset{bc>0}{\iff} ab + bc < ac + bc \iff b(a+c) < c(a+b) \overset{a>0}{\iff} \\
        ab(a+c) < ac(a+b) \overset{a+c>0, \; a+b>0}{\iff} \frac{ab}{a+b} < \frac{ac}{a+c} \equiv a\otimes^\ASD b < a\otimes^\ASD c
    \end{multline*}
\finex
\end{proof}

\begin{lemma}\label{distribute}
    $\otimes^\ASD$ is left and right distributive over $\oplus^\ASD$.
\end{lemma}
\begin{proof}
    Let $a,b,c \in \ASD$:
    By induction on the order relation between $b,c$:
    \begin{itemize}    
    \item $b=c$: \\
        Trivial.
    \item $b<c$: \\
    \begin{equation}\label{eq:proof_right}
        a \otimes (b \oplus c) = \frac{a\ \textrm{min}(b, c)}{a + \textrm{min}(b, c)} \overset{\text{by IH}}{=} \frac{ab}{a + b} = a \otimes b
    \end{equation}
    \begin{equation}\label{eq:proof_left}
        (a \otimes b) \oplus (a \otimes c) = \textrm{min}((a \otimes b), \; (a \otimes c))  \overset{\text{by Lemma~\ref{inequality} and IH}}{=\joinrel=\joinrel=} a \otimes b
    \end{equation}
    Then our proof is completed by using reflexivity on \ref{eq:proof_right} and \ref{eq:proof_left}. \\
    \item $b>c$: \\
        Similar to $b<c$.
    \end{itemize}  
\noindent
The right side distributivity follows by the commutativity properties of the operations.
\finex
\end{proof}




\begin{theorem}\label{th:semiring}
    $(dom(\ASD), \oplus^{\ASD},\otimes^{\ASD})$ is a semiring.
\end{theorem}

\begin{proof}
    \noindent
We have all the elements to conclude the proof:
    \begin{itemize}
        \item $(dom(\ASD), \oplus^{\ASD})$ is associative and satisfies the closure property (by Lemma~\ref{com_plus});
        \item $(dom(\ASD), \otimes^{\ASD})$ is associative and satisfies the closure property (by Lemma~\ref{com_times});
        \item Multiplication distributes over addition (by Lemma~\ref{distribute})
    \end{itemize}    

\finex
\end{proof}

When dealing with multiple equal durations in music, it is practical to extend the $\oplus$ operation to define a scalar multiplication. For a duration $a 
\in \ASD$ and a scalar $n \in \mathbb{Q}$, it is denoted by the function $\textrm{repeat}^\ASD(a,n) = a/n$.


    


\subsection{Relative Symbolic Durations}

We define the semiring of Relative Symbolic Duration units $\RSD^\delta= (\mathbb{Q}^+\cup\{ 0 \}, \oplus, \otimes)$ relative to the reference duration $\delta$.

\subsubsection{The domain}
The domain $dom(\RSD^\delta) = \mathbb{Q}^+\cup\{ 0 \}$ of $\RSD$ contains durations measured relative to a reference duration value. Intuitively, smaller values correspond to shorter notes. The limiting case is the duration 0, which can be used, for example, for grace notes. 
We define $\prec^{\RSD}$ to be the strict order of absolute musical durations on the domain of $\ASD$. Elements of $\RSD$ are defined such that $\forall a,b \in dom(\RSD), a \prec^{\RSD} b \iff a < b$.

\subsubsection{Operations}\label{sec:rsd}
Similarly to $\ASD$, $\oplus^{\RSD} \equiv \mathrm{max}$ is used to select the larger duration and $\otimes^{\RSD} \equiv +$ is used to add two durations together. 
The repeat operation can be defined as $\textrm{repeat}^\RSD(a,n) = a*n$.

We skip the proof of $\RSD$ being a semiring for brevity. It can also be noted that the operations and domain we defined are equivalent to those of a tropical semiring~\cite{pin1998tropical}, so the proof for tropical semirings is also valid for our case.
Differently from a tropical semiring, however, we don't have the absorption property of the $\otimes$ neutral elements $\zero$, i.e., $\forall x\in \RSD$, $\zero \otimes x = x \otimes \zero = \zero$. In order to verify this, we would need to swap the min with the max (and vice-versa) for the $\otimes$ in our two semirings, but this would make for a non-musically useful operation, violating the ultimate objective of this research.



\subsection{A General Duration Framework}\label{sec:general}
Table~\ref{tab:def_table} summarises our formalization of ASD and RSD units. In the following, we introduce a morphing function to convert between these two units. Finally, we include in our framework the duration modifiers that are used in musical scores, i.e., ties, dots, and tuplets.

\subsubsection{Morphing between time units}\label{subsec:morphing}\hfill\break
Given a reference duration value $\delta$, we define the reciprocal function $f(x)=\delta/x$ that maps every element $x \in dom(\ASD)$ to its correspondent in $\RSD^\delta$, and vice-versa. It is trivial to see that this function is isomorphic and order-preserving (it preserves the ordering in the respective source/target domains, even though the order in $\ASD$ is reversed with respect to $\RSD$);
it follows that $f$ is a \textit{Homomorphism}, i.e. $\forall a, b \in dom(\RSD), \; f(a\otimes^{\RSD} b) = f(a) \otimes^{\ASD} f(b)$. 
The choice of $\delta$ has interesting practical implications. For example, by setting it to a beat duration (which depends on the time signature), we obtain units typically used in music research to reduce the dependency on the time signature. By setting it to a
quarter note duration we obtain the so-called \textit{quarter length} durations, commonly used for general applications since they do not depend on other score parameters.


\begin{table}[]
    \begin{center}
    \begin{tabular}{l|c c c}
        & $\mathbb{S}$ & $\oplus$ & $\otimes$  \\
        \hline
        ASD $\ASD$ & $\mathbb{Q}^+\cup \{ \infty \} $ &$\; \textrm{min}(a, b) \; $ & $\; \frac{ab}{a+b} \; $\\
        RSD $\RSD$ & $\mathbb{Q}^+\cup\{ 0 \} $ & $\; \textrm{max}(a, b) \;$ & $ \; a+b \; $\\
    \end{tabular}
    \end{center}
    \caption{Table comparing the semirings of Absolute and Relative Symbolic Durations.}
    \label{tab:def_table}
\end{table}


\subsubsection{Duration modifiers}\label{subsubsec:mus_concepts}\hfill\break
In a musical score, there are some graphical symbols, i.e., ties, dots, and tuplet groupings, that modify the duration of the notes/rests they are assigned to. In this section, we will define ties, dots, and tuplets as functions applied to elements of either $\ASD$ or $\RSD$. 
We use the symbol $\mathcal{X}$ to refer to either of the structures $\ASD$ or $\RSD$.


First, let us consider the ties between notes. The total duration of two tied notes $a, b \in \mathcal{X}$ can be easily captured by the $\otimes$ operation.

\begin{definition}
The total duration of two tied notes $a, b \in \mathcal{X}$ is given by function $\textrm{tie}: \mathcal{X} \times \mathcal{X} \to \mathcal{X}$
    \begin{equation}\label{eq:tie}
        \mathrm{tie}(a, b) = a \otimes^\mathcal{X} b
    \end{equation}
\end{definition}

Another musical concept that can prolong the duration of a musical note is the dot. A dotted note $a$ can be seen as a function $\textrm{dot}$ applied to the note $a$. This can be generalized for an arbitrary number of dots:
\begin{definition}
    The function $\textrm{dot}: \mathcal{X} \times \mathbb{N} \to \mathcal{X}$ applied to a note $a \in \mathcal{X}$ is inductively defined as follows:
\begin{align}
    \mathrm{dot}(a, 0) &= a \\
    \mathrm{dot}(a, n+1) &= \mathrm{repeat}(a, \frac{1}{2^{n+1}}) \otimes^\mathcal{X} \mathrm{dot}(a, n)
\end{align}   
\end{definition}

Another function that can be used to construct musical duration is the tuplet function. The duration of a note in a tuplet of total duration $a\in\mathcal{X}$ can be seen as a function with two parameters, the base note duration $a$ and the type of tuple $\gamma$ (in this case 3 for triplet).
\begin{definition}\label{def:tuplet}
    Let $a \in \mathcal{X}$, $\gamma \in \mathbb{N}_{>2}$. The tuplet function,$t: \mathcal{X} \times \mathbb{N}^* \to \mathcal{X}$, is defined as:
\begin{equation}
    t(a, \gamma) = 
    \mathrm{repeat}(a,  \frac{2}{\gamma })
\end{equation}
\end{definition}






\begin{example}
We use the formalisms introduced in this section on the problem of Example~\ref{ex:running}, where the goal was to import the score from an MEI file, split the third note (the C) in the top voice into two, and produce a pianoroll representation. This process can abstractly be described solved as:
(1) read all durations $[d_1,d_2,\dots,d_n]$ from the input MEI file; (2) compute the values of the notes under the triplet with Definition~\ref{def:tuplet}; (3) split the third note into two notes with duration $d_{\textrm{new}} = \textrm{repeat}(d_3,1/2)$; (4) find each note onset and offset position by $[d_1 \otimes d_2 \otimes \dots \otimes d_n]$;  (5) for the last note offset of each voice, compute the maximum with the $\oplus$ operator; (6) output the pianoroll representations for the two voices, using the start times and durations thus calculated.
\end{example}

\section{From Abstract Description to Algorithm Implementation}\label{sec:practical_usage}
In the previous section, we introduced an abstract formalism to describe algorithms on music sequences. We now discuss cases where it is more efficient to perform such algorithms in ASD units or in RSD units.

\subsection{Advantages and Disadvantages}
The use of ASD units can bring advantages in terms of data types because it can give a prevalence of integers over floating point (and periodic) values. For this to be the case, we need to deal with durations that span a maximum of a whole note. In a 4/4 piece, this will correspond to durations of one measure.
This does not mean that algorithms implemented in ASD units cannot handle multiple measures, but rather that they should follow a ``divide et impera'' principle where every measure is handled independently. This is already quite common in file-parsing systems since scores are encoded measure by measure in file formats such as MusicXML and MEI.

In terms of the number of computations, ASD units are ideal for applications that concern the graphical symbols used in the score, for example, changing the pitch of a note, changing a duration, or segmenting a musical score. Such applications can skip the costly computation of common divisor, and conversion to RSD units altogether. 
Instead, when the measure is not specified (which could be the case, for example, in handling a MIDI file), or when we want to do operations that don't follow the measure segmentation (e.g., segmenting a score between measures), the usage of RSD units is preferred.

\begin{example}
    Let us consider the problem of Example~\ref{ex:running}. By considering ASD units, we can parse the input score file simply by copying the values for the note graphical durations. The splitting of the third note of duration $16$ in two parts yields two notes of duration $16*2 = 32$.
\end{example}

A big limit in the efficiency of ASD units is posed by time signatures where the beat is a dotted note, for example, 6/8. A dotted note will make the duration assume noninteger, or even periodic, values. 
A possible solution to this problem is given in the next section.

\subsection{The lazy evaluation case}
It is common for systems that deal with musical scores to have a generic import function, where the score file is converted to some internal representation. If in In this step, the user did not yet specify the set of operations they intend to perform, the choice of whether to use ASD or RSD cannot be performed made.
In order to let the system choose between ASD and RSD to exploit the advantages described in the previous section, we suggest using a \textit{lazy evaluation} parsing strategy. First, we propose to reduce the domain by considering only the \textit{ASD} values $\{2^{n} \mid \forall n \in [0..7] \}$ (i.e. only single graphical note/rest symbols). Duration modifiers such as dots or tuplets are imported as functions $\textrm{dot}$ or $\textrm{tuple}$ without  being computed.
Only when the user specifies a task, 
will these functions be resolved to actual values, and the task is performed either in ASD or RSD units, depending on what would allow for the most efficient computation.
From a functional programming perspective, this can be viewed as a Monad transformation~\cite{wadler1992essence} of the parsed elements. 




\subsection{Implementation}
We provide a proof of concept of the practical utility of the methods introduced in this paper, by implementing some functions in the Python library \textsf{Partitura}~\cite{partitura_mec}. 
The core of this library, i.e., the \textit{Timeline} object, uses RSD units, in particular on the common divisor technique described in the introduction. However, some functions in the file parsing module are modular enough to make it possible to run them in ADS without the need of making major changes to the rest of the library. These are: (1) the functions to compute the common divisor for integer encoding of RSD durations, (2) the function to find the longest voice in a measure, and (3) the computation of the actual duration for a note inside a tuplet.
We also implement an alternative (still partial) parser of **Kern files that leverages a lazy evaluation approach.

\section{Conclusions and Discussion}\label{sec:discussion}
In this paper, we proposed an alternative approach to handling the symbolic music durations from musical scores, that is based on absolute symbolic duration (ASD) units. We formalized ASD, and the (typically used) relative symbolic duration (RSD) units, in a single mathematical framework, and paired them with two operations. The result is two semirings: algebraic structures that enable an abstract description of algorithms on symbolic durations.
We then moved to a more practical discussion and described some use cases where one unit is more efficient than the other, in terms of data types (integers vs floating point) and number of calculations. Finally, we advocated a functional parsing of symbolic music formats that can select the most efficient way of performing the various operations in an algorithm and enable considerable speed-up for common use cases.

It is clear that the proposals in this paper are mostly of theoretical interest, and belong to the research branch that formalizes musical elements with mathematical structures~\cite{10.1007/978-3-030-21392-3_6, mazzola2012topos, popoff2016k}. However, our interest in this topic started from our practical experience with parsing and processing musical score files to use their information as input for music information retrieval (MIR) systems. While the improvement in efficiency that our methods may enable is negligible for a single score, large deep-learning models have to load thousands of scores, thus making each small optimization much more useful. For example, we will probably soon see some general tokenization techniques for musical scores (similar to the multiple ones that have been proposed for MIDI files~\cite{miditok2021}); in this context, a tokenization that focuses on the graphical symbols using only ASD, could enable major speedups in computing time. 



\section{Acknowledgements}
This work was supported by the European Research Council (ERC) under the EU's Horizon 2020 research \& innovation programme, grant agreement No.\ 101019375 \\(\textit{Whither Music?}), the Federal State of Upper Austria (LIT AI Lab), and JSPS Kaken 20H04302, 21H03572.

\printbibliography
\end{document}